%% file: main.tex
\documentclass[11pt]{article}
\usepackage[margin=1in]{geometry}
\usepackage{algorithm}
\usepackage{amsmath,mathtools}
\usepackage{color}
\usepackage{amsfonts}
\usepackage{amssymb}
\usepackage{xspace}
\usepackage{comment}
\usepackage{hyperref}
\usepackage{float}
\usepackage{amsmath}
\usepackage{algpseudocode}

\usepackage{subfigure}
\usepackage{times}
\usepackage{sidecap}
\usepackage{enumerate}
\usepackage{tcolorbox}

\usepackage{blindtext}
\usepackage{enumitem}
\usepackage{xcolor}
\def \polylog{\operatorname{polylog}}
\usepackage[symbol]{footmisc}
\usepackage{varioref}
\usepackage{array}
\newcolumntype{P}[1]{>{\centering\arraybackslash}p{#1}}
\newcolumntype{M}[1]{>{\centering\arraybackslash}m{#1}}

\usepackage{amsmath,amssymb, amsthm, enumerate,boxedminipage}

\theoremstyle{definition}
\newtheorem*{def*}{Definition}
\newtheorem{definition}{Definition}
\newtheorem{theorem}{Theorem}[section]
\theoremstyle{definition}

\newtheorem{lemma}{Lemma}

\usepackage{mathtools}

\theoremstyle{definition}
\newtheorem*{prf*}{Proof}
  
\theoremstyle{definition}
\newtheorem*{prfthm*}{Proof of Theorem}

\newtheorem{remark}{Remark}

\begin{document}
\title{Improved Deterministic Leader Election in Diameter-Two Networks\thanks{The work of A. R. Molla and S. Sivasubramaniam were supported, in part, by ISI DCSW/TAC Project (file no. G5446 and G5719).}}

\author{Manish Kumar \thanks{Indian Statistical Institute, Kolkata 700108, India.\hbox{E-mail}:~{\tt manishsky27@gmail.com}.}
\and Anisur Rahaman Molla \thanks{Indian Statistical Institute, Kolkata 700108, India.  \hbox{E-mail}:~{\tt anisurpm@gmail.com}.} \and Sumathi Sivasubramaniam \thanks{Indian Statistical Institute, Kolkata 700108, India. \hbox{E-mail}:~{\tt sumathivel89@gmail.com}.}}
\date{}

\maketitle \thispagestyle{empty}

\maketitle

\begin{abstract}
In this paper, we investigate the leader election problem in diameter-two networks. Recently, Chatterjee et al. [DC 2020] studied the leader election in diameter-two networks. They presented a $O(\log n)$-round deterministic {implicit} leader election algorithm which incurs optimal $O(n\log n)$ messages, but a drawback of their algorithm is that it requires knowledge of $n$. An important question-- whether it is possible to remove the assumption on the knowledge of $n$ was left open in their paper. Another interesting open question raised in their paper is whether {\em explicit} leader election can be solved in $\Tilde{O}(n)$ messages deterministically. In this paper, we give an affirmative answer to them. Further, we solve the {\em broadcast problem}, another fundamental problem in distributed computing, deterministically in diameter-two networks with $\Tilde{O}(n)$ messages and $\tilde{O}(1)$ rounds without the knowledge of $n$. In fact, we address all the open questions raised by Chatterjee et al. for the deterministic leader election problem in diameter-two networks. In particular, our results are:
\begin{enumerate}
    \item We present a deterministic {\em explicit} leader election algorithm which takes $O(\log \Delta)$ rounds and $O(n \log \Delta)$ messages, where $n$ in the number of nodes and $\Delta$ is the maximum degree of the network. The algorithm works without the knowledge of $n$. The message bound is tight due to the matching lower bound, showed Chatterjee et al. [DC 2020]. 
    
    \item We show that {\em broadcast} can be solved deterministically in $O(\log \Delta)$ rounds using $O(n \log \Delta)$ messages. More precisely, a broadcast tree can be  computed with the same complexities and the depth of the tree is $O(\log \Delta)$. This also doesn't require the knowledge of $n$. 
\end{enumerate} 
To the best of our knowledge, this is the first $\Tilde{O}(n)$ deterministic result for the explicit leader election in the diameter-two networks, that too without the knowledge of $n$. 
\end{abstract}
\input{introduction}

\input{model}

\input{related_work}

\input{deterministic}

%\input{randomized}
\input{conclusion}

\bibliographystyle{plainurl}
\bibliography{reference}
\end{document}

%% file: introduction.tex
\section{Introduction}\label{sec: introduction}
In the four decades since its inception, leader election has remained a well explored and fundamental problem in distributed networks \cite{LSP82, Lann'77, PSL80}. The basic premise of leader election is simple: given a group of $n$ nodes, a unique node is elected as a leader (where $n$ denotes the number of nodes in the network).  Depending on the nodes knowledge of the leader, there are two popular versions. In the first version (known as the {\em implicit} leader election), the non-leader nodes are not required to know the leader's identity; it is enough for them to know that they are not the leader. The {\em implicit} leader election is quite well studied in literature \cite{AW04, KM21, KPP0T15, KPPRT15, L96}. In the other version (known as {\em explicit} leader election), the non-leader nodes are required to learn the leader's identity. The implicit version of  the leader election is the generalized version of the (explicit) leader election.  Clearly, there is a lower bound of  $\Omega(n)$ for message complexity in the explicit version of the problem. In this paper, we study the explicit version of the problem. In particular, we show an improvement on the existing deterministic solution for the implicit leader election algorithm presented in~\cite{CPR20} and provide an algorithm for turning the implicit leader election explicit without any additional overhead on messages.

Leader election has been studied extensively with respect to both message and round complexity in various graph structures like rings \cite{Lann'77,YM18}, complete graphs \cite{AG91, AMP18, Humblet84, KKM90, KMZ87, KMZ89, Singh91}, diameter-two networks \cite{CPR20} etc., as well as in general graphs \cite{GHS83, GRS18, KPPRT15, L96, Peleg90} \footnote{We interchangeably  use the word "graph" and "network" throughout the paper.}. Earlier works were primarily focused on providing deterministic solutions. However, eventually, randomized algorithms were explored to reduce mainly the message complexity (see \cite{AMP18, GRS18, KPP0T15, KPPRT15} and the references there in). Kutten et al. gave the fundamental lower bound for leader election in general graphs with $\Omega(m)$ message complexity and $\Omega(D)$ round complexity \cite{KPP0T15}, where $m$ is the number of edges and $D$ is the diameter of the graph. This bound is applicable for all graphs with diameters greater than two, whether the algorithm is deterministic or randomized. For the clique, recently a tight message lower bound of $\Omega(n \log n)$ is established by Kutten et al. \cite{KRTZ23} for the deterministic algorithms under simultaneous wake-up of the nodes. The same lower bound was shown earlier by Afek and Gafni (1991) \cite{AG91}, but assumes adversarial wake-up. Table~\ref{tbl: deterministic_list} presents an overview of the results (deterministic). Recently, diameter-two networks were explored, and the message complexity was settled by providing a deterministic algorithm with $O(n \log n)$ message complexity \cite{CPR20}. %\anis{which work are we talking in this last line?} \manish{It is Chatterjee et al. \cite{CPR20} but seems inappropriate here. Should we remove this line?}

Our work is closely related to the work by Chatterjee et al. \cite{CPR20}. In their work, the authors studied leader election (the {\em implicit} version) in diameter-two networks. They presented a deterministic algorithm with $O(n \log n)$ message complexity and $O(\log n)$ round complexity. Crucially, their algorithm requires prior knowledge on the size of the network, $n$. In comparison to this, our algorithm elects a leader {\em explicitly} without prior knowledge of $n$. Our algorithm uses $O(n \log \Delta)$ messages and finishes in $O(\log \Delta)$ rounds, where $\Delta$ is the maximum degree of the graph (see, Table~\ref{tbl:comparision}). In addition to this, we show how to leverage the edges used during the leader election protocol to create a broadcast tree for the diameter-two graphs with a message and round complexity of $O(n \log \Delta)$ and $O(\log \Delta)$, respectively. Computing a broadcast tree efficiently is another fundamental problem in distributed computing. A broadcast tree can be used as a subroutine to many distributed algorithms which look for message efficiency. Finding a deterministic $\Tilde{O}(n)$-message and $\Tilde{O}(1)$-round broadcast algorithm in diameter-two networks was also left open in \cite{CPR20}. We have addressed it.

\begin{table*}[tbh]
%\centering
\begin{tabular}{|P{4.00cm}|P{3.50cm}|P{3.50cm}|P{3.80cm}|}
%\begin{tabular}{|P{3.50cm}|P{3.20cm}|P{3.00cm}|P{2.00cm}|}
\hline
\multicolumn{4}{|c|} {\sc Deterministic (Explicit) Leader Election Results}\\
\hline
 Paper & Message Complexity & Round Complexity  &  Graph of Diameter  \\
\hline
&&&\\
Afek-Gafni \cite{AG91} & $O(n \log n)$& $O(\log n)$ * & $D=1$\\
     &&&\\
 Kutten et al. \cite{KRTZ23} & $\Omega(n \log n)$& $\Omega(1)$& $D=1$\\
     &&&\\
\textbf{This paper} & $O(n \log \Delta)$& $O(\log \Delta)$& $D=2$\\
     &&&\\
Chatterjee et al. \cite{CPR20}   & $\Omega(n \log n)$& $\Omega(1)$ **& $D=2$\\
     &&&\\
Kutten et al. \cite{KPP0T15} & $O(m \log n)$& $O(D \log n)$&$D\geq 3$ \\
     &&&\\
Kutten et al. \cite{KPP0T15} & $\Omega(m)$& $\Omega(D)$& $D\geq 3$\\
     &&&\\
\hline
\end{tabular}
    \caption{Best known deterministic leader election results on networks with different diameters. Since $\Delta = \Omega(\sqrt{n})$ in diameter-two graphs, $\log \Delta = O(\log n)$, see the Remark~\ref{rem:lb-on-max-degree} below. So our upper bound doesn't violate the message lower bound in \cite{CPR20}. * Attaining $O(1)$ time requires $\Omega(n^{1+\Omega(1)}$ ) messages in cliques, whereas achieving $O(n \log n)$
messages requires $\Omega(\log n)$ rounds; see \cite{AG91}. ** $\Omega(1)$ is a trivial lower bound.} %In upper bound complexity, round complexity is w.r.t. best known message complexity.}
\label{tbl: deterministic_list}
\end{table*}

\textbf{Paper Organization:} In the rest of this section~\ref{sec: introduction}, we state our results. In Section~\ref{sec: model}, we present our model and definitions.  We briefly introduce various related works in  Section~\ref{sec: related_work}. We present our algorithms for deterministic leader election and broadcast tree formation in Section~\ref{sec: deterministic_leader}. And finally, we conclude in Section~\ref{sec: conclusion} with several open problems.

\subsection{Our Results}\label{sec: results}
Our work focuses on the deterministic leader election in diameter-two networks without the knowledge of number of nodes. Apart from this, by leveraging the  leader election protocol, we show that \textit{broadcast} can be solved deterministically, matching the complexity of the leader election algorithm. Specifically, we have the following results.
\begin{enumerate}
    \item We present a deterministic {\em explicit} leader election algorithm which takes $O(\log \Delta)$ rounds and $O(n \log \Delta)$ messages, where $n$ in the number of nodes and $\Delta$ is the maximum degree of the network. The algorithm works without the knowledge of $n$. The message bound is tight due to the matching lower bound, showed by Chatterjee et al. in \cite{CPR20}.
    \item We show that {\em broadcast} can be solved deterministically in $O(\log \Delta)$ rounds using $O(n \log \Delta)$ messages. More precisely, we show that a broadcast tree, of depth at most $O(\log \Delta)$ can be  computed with the same complexities.
\end{enumerate}

%% file: model.tex
\section{Model and Definition}\label{sec: model}
Our model is similar to the one in  \cite{CPR20}. We consider the distributed network to be an undirected graph $G = (V, E)$ of $n$ nodes and diameter $D=2$.  Each node has a unique ID of size $O(\log n)$ bits. The model is a \textit{clean network model}  in the sense that the nodes are unaware of their neighbors' IDs initially, also known as $KT_0$ model \cite{Pelege2000}. 
%\anis{we don't need this sentence. I am removing it} We consider the graph to be connected (and of diameter 2), therefore, $n\leq m$ and $m < \frac{n(n-1)}{2}$, where $m = |E|$ and $n=|V|$.
The network is synchronous. The nodes communicate via passing messages in a synchronous round. We limit each message to be of size at most $O(\log n)$ bits as in the CONGEST communication model in distributed networks \cite{Pelege2000}. In each round, nodes may send messages, receive messages and perform some local computation. The round complexity of an algorithm is the total number of rounds of communication taken by the algorithm before termination. The message complexity is the total number of messages exchanged in the network throughout the execution of the algorithm.  Throughout this paper, we assume that all nodes are awake initially and simultaneously start executing the algorithm.

We will now formally define the implicit and explicit version of leader election in our model.
\begin{definition}[\textbf{Implicit Leader Election}]
Consider an $n$-node distributed network. Let each node maintain a state variable that can be set to a value in $\{\perp,\, NONELECTED, \, ELECTED\}$, where $\perp$ denotes the `undecided' state. Initially, all nodes set their state to $\perp$. In the implicit version of leader election, it requires that exactly one node has its state variable set to $ELECTED$ and all other nodes are in state $NONELECTED$. The unique node whose state is $ELECTED$ is the leader.  
\end{definition}

\begin{definition}[\textbf{Explicit Leader Election}]
Consider an $n$-node distributed network. Let each node maintain a state variable that can be set to a value in $\{\perp,\, NONELECTED, \, ELECTED\}$, where $\perp$ denotes the `undecided' state. Initially, all nodes set their state to $\perp$. In the explicit version of leader election, it requires that exactly one node has its state variable set to $ELECTED$ and all other nodes are in state $NONELECTED$. Further, the $NONELECTED$ nodes must know the identity of the node, whose state is  $ELECTED$, the leader.
% Consider an $n$-node network of diameter two graph. The explicit leader election requires electing a leader from the set of $n$-nodes such that every node maintain a state variable that can set to a value in $\{\perp,\, NONELECTED, \, ELECTED\}$, where $\perp$ denotes the `undecided' state; initially, state is set to $\perp$. Finally, exactly one node is $ELECTED$ rest other nodes are $NONELECTED$ and these $NONELECTED$ nodes know about the $ELECTED$ one.
\end{definition}

%% file: related_work.tex
\section{Related Work}\label{sec: related_work}
\begin{table*}[t]
%\centering
\begin{tabular}{|P{3.10cm}|P{3.50cm}|P{3.00cm}|P{1.80cm}|P{3.00cm}|}
%\begin{tabular}{|P{3.50cm}|P{2.20cm}|P{2.00cm}|P{2.00cm}|P{2.00cm}|}
\hline
\multicolumn{5}{|c|} {\sc Deterministic Leader Election in Diameter-Two Graphs}\\
\hline
 Paper & Message Complexity & Round Complexity  &  Type & Knowledge of $n$ \\
\hline
&&&&\\
Chatterjee et al. \cite{CPR20}& $O(n \log n)$& $O(\log n)$              & Implicit & YES\\
&&&&\\
    
% Chatterjee et al. \cite{CPR20} \manish{Is this row, Okay?} & $\Omega(n \log n)$      &$\Omega(1) \manish{the trivil lower bound.}$& Implicit/Explicit& $-$\\
% &&&&\\
    
This paper & $O(n \log \Delta)$& $O(\log \Delta)$& Explicit& NO\\
     &&&&\\
% This paper \manish{!!} & $\Omega(n \log n)$& $\Omega(1)$& Implicit/Explicit\\
%      &&&\\
\hline
\end{tabular}
    \caption{Comparison of the current paper to the state-of-the-art.}
\label{tbl:comparision}
\end{table*}
% \begin{table}[H]
% \footnotesize
% \begin{table}[]{|P{3.5cm}|P{2cm}|P{2cm}|P{2cm}|}
% \hline
% \multicolumn{4}{|c|} {Leader Election}\\
% \hline
%   & Message Complexity & Round Complexity & Type\\
% \hline
%     &&&\\
%     Chatterjee et al. \cite{CPR20}& $O(n \log n)$& $O(\log n)$              & Implicit \\
%     &&&\\
    
%       Chatterjee et al. \cite{CPR20} & $\Omega(n \log n)$      &$-$& Implicit\\
%     &&&\\
    
%     Our Work & $O(n \log \Delta)$& $O(\log \Delta)$& Explicit\\
%     &&&\\
% \hline

% \end{tabular}
%     \caption{Summary of the leader election algorithm in deterministic setting w.r.t. diameter-two}
% \label{tbl:comparision}
%\end{table}
In 1977, the leader election problem was introduced by Le Lann in the ring network \cite{Lann'77}. Since then the problem has been studied extensively in different settings. The leader election problem has been explored in both implicit and explicit versions over the years~\cite{KKMPT12,L96,Peleg90,KPPRT15, GRS18, KM021} for a variety of models and settings, and for various graph topologies such as cliques, cycles, mesh, etc., (see~\cite{KKMPT12, KM21, KKKSS08, KSSV06, Peleg90, S07, Tel94,refai2010leader} and the references therein for more details). In general, the implicit leader election suffices for most networks.

%Both implicit and explicit versions of the Leader election problem have been well studied over the years~\cite{KKMPT12,L96,Peleg90,KPPRT15} for a variety of different models and communication, and for various typologies such as a cycle, mesh, etc., (see~\cite{KKMPT12,L96, Peleg90, S07, Tel94,refai2010leader} and the references there in for more details). Generally speaking, the implicit variation of the leader election algorithm is enough of a guarantee for most networks.

Both deterministic and randomized solutions exist for leader election. For the randomized case, for complete graphs, Kutten et al.~\cite{KPPRT15} showed that $\Tilde{\Theta}(\sqrt{n})$ is a tight message complexity bound for randomized (implicit) leader election. For any graph with diameter greater than 2, the authors in~\cite{KPPRT15} showed that $\Omega(D)$ is a lower bound for the number of rounds for leader election using a randomized algorithm (they also showed a lower bound for the message complexity, $\Omega(m)$). Recently, Chatterjee et al.,~\cite{CPR20} showed a lower bound of $\Omega(n)$ for the message complexity of randomized leader election in diameter-two graphs.

In the deterministic case, it is known that $\Theta(n\log n)$ is a tight bound on the message complexity for complete graphs~\cite{AG91,KRTZ23}. This tight bound also carries over to the general case as seen from~\cite{AG91,KMZ84,KMZ89}. In our work, we restrict our model to graphs of diameter-two. For diameter-two graphs, Chatterjee and colleagues  provide a $O(\log n)$ round algorithm that uses $O(n \log n)$ messages. However, their algorithm requires knowledge of $n$, our algorithm provides an algorithm that requires no prior knowledge of $n$ and runs in $O(\log n)$ rounds with $O(n\log n)$ message complexity.

%% file: deterministic.tex
\section{Deterministic Leader Election in Diameter-Two Networks}\label{sec: deterministic_leader}
We present a deterministic (explicit) leader election algorithm for diameter-two networks with $n$ nodes in which the value of $n$ is unknown to nodes in the network. In this section, we answer several questions raised in \cite{CPR20}. Specifically, we address the following: (i)  Can explicit leader election be performed in  $\Tilde{O}(n)$ messages in diameter-two graphs deterministically? (ii) Given the leader election algorithm, can broadcast can be solved deterministically in diameter-two graphs with $\Tilde{O}(n)$ message complexity and $O(\polylog n)$ rounds if $n$ is known, and crucially (iii) "Removing the assumption of the knowledge of n (or showing that it is not possible) for deterministic, implicit leader election algorithms with  $\Tilde{O}(n)$ message complexity and running in  $\Tilde{O}(1)$ rounds is open as well."  In this section, we solve the explicit leader election with $\Tilde{O}(n)$ message complexity, along with that our algorithm solves the explicit leader election without the knowledge of $n$; thus addressing the questions (i) and (iii). We further present a solution for the question~(ii) that too without the knowledge of $n$. 

\subsection{Algorithm}
Our algorithm is inspired from the work done by Chatterjee et al. \cite{CPR20}. They presented an algorithm for implicit leader election that ran in  $O(\log n)$ rounds with $O(n \log n)$ message complexity (with the knowledge of $n$). Our Algorithm~\ref{alg : deterministic}, achieves somewhat better result without the knowledge of $n$ and also elects the leader explicitly. 

As mentioned earlier (in Section~\ref{sec: model}), each node has a unique ID. For any node $v \in V$, let's denote the degree of $v$ by $d_v$ and the ID of $v$ by $ID_v$. The \emph{priority}  $\mathcal{P}_v$, of node $v$, is a combination of the degree and ID of the node $v$ such that $\mathcal{P}_v = \langle d_v, ID_v \rangle$. The leader is elected based on the priority, which is decided by the degree of the node. In the case of a tie, the higher ID gets the higher priority. Essentially, the node with the highest priority becomes the leader.
%At any round, $L_v$ denotes the highest priority that $v$ has so far learned (among all the probe messages it has received, in the current round or in some previous rounds). $v$ also keeps track of the neighbor who informed it about the highest priority node in $N_v$.

Our algorithm runs in two phases of $O(\log d_v)$ rounds each. In the first $O(\log d_v)$ rounds, we eliminate as many invalid candidates as possible. In the second phase, all candidates except the actual leader are also eliminated, culminating in the election of a unique leader.

\medskip

\noindent\textbf{Detailed description of the algorithm:}

\medskip
\noindent Initially, every node is a "candidate" and has an “active” status. Each node $v$ numbers its neighbors from $1$ to $d_v$ arbitrarily, denoted by $w_{v,1}, w_{v,2}, \dots, w_{v,d_v}$. For the first $i=1$ to $\log d_v$ rounds, if $v$ is active,  then node $v$ sends a message containing $\mathcal{P}_v$ to its neighbors $w_{v,2^{i-1}}, \cdots, w_{v,\text{min}\{d_v,2^i-1\}}$ . If $v$ encounters a priority higher than its own from its neighbors (either because a neighbor has a higher priority or has heard of a node with higher priority) then $v$ becomes "inactive" and "non-candidate". That is, $v$ does not send any further messages to its neighbors containing $v$'s priority. Although, $v$ may send higher priority message based on the received message's priority (explained later). Let $L_v$ denotes the ID of the current highest priority node known to $v$.  At the beginning of the execution, $L_v$ is simply $\mathcal{P}_v$. If at the end of the first $\log d_v$ rounds, $L_v=\mathcal{P}_v$ then $v$ declares itself leader temporarily. Further, $v$ waits for $\log d_v$ rounds. If at the end of $\log d_v$ rounds $v$ is still the candidate node ($v$ has not heard from a node about the higher priority) then $v$ becomes the leader.
%\medskip

%\noindent\textbf{Detailed description of the algorithm:}
There are two major phases to the algorithm. For the first $\log d_v$ rounds, we eliminate as many invalid candidates (the node which has encountered higher priority node) as possible, as follows. $N_v$ contains the ID of the neighbor that informed $v$ about the current highest priority. As mentioned before, $L_v$ contains the current highest priority known to $v$. Let  $\chi_v$ denote the (possibly empty) set of $v$'s neighbors from whom $v$ has received messages in a round during this phase, and $\mathcal{P}(\chi_v)$ be the set of $\mathcal{P}$s sent to $v$ by the members of $\chi_v$ such that $\mathcal{P}_u$ be the highest $\mathcal{P}$ in $\mathcal{P}(\chi_v)$. If $\mathcal{P}_u$ is higher than that of $L_v$ then $v$ stores the highest priority seen so far in $L_v$.  Further, $v$ informs $N_v$ about $L_v = \mathcal{P}_u$, i.e., the highest $\mathcal{P}$ it has seen so far. This particular step exploits the neighborhood intersection property to ensure that information about higher priority nodes is disseminated quickly. Then $v$ updates $N_v$. Finally, $v$ tells every member of $\chi_v$ about $L_v$, i.e., the highest $\mathcal{P}$ it has seen so far. If $L_v \neq P_v$ then $v$ becomes “inactive” and “non-candidate”. Notice that an "inactive" and "non-candidate" node $v$ only disseminates the information of higher priorities it hears, to $N_v$.

At the end of the first $\log d_v$ rounds, we begin the final phase of the election. If $v$ is still the candidate node then $v$ waits for $ \log (d_v)$ rounds. Furthermore, if $v$ does not receive any higher priority message then $v$ declares itself as the leader and informs its neighbors. Then each neighbor of $v$, say $u$, informs their neighbors about the election of $v$ via set of $\Psi_u$ nodes. On the other hand, if there exists a node whose priority is higher than the priority of $v$ then $v$ gets to know about the leader and informs all the nodes to whom $v$ has communicated (so far) about the leader's ID (that is the set $\Psi_v$) and exits. Hence, All the nodes elect the same leader whose priority is the highest. Our claim is that given certain properties of the degree (see Lemma~\ref{lem: size}) we can guarantee that the second phase of waiting for $\log \Delta$ rounds eliminates all but a unique candidate, which then becomes leader.

Now, we would discuss some important lemmas and the correctness of the algorithm. Finally, we conclude the result in Theorem~\ref{thm: deterministic}.

\begin{figure}

\begin{algorithm}[H]\caption{\sc Deterministic-Leader-Election: Code for a node $v$}\label{alg : deterministic}
\begin{algorithmic}[1]
\Require{A two diameter connected anonymous network. Each node possess unique ID.}
\Ensure{Leader Election.}

\Statex
\State $v$ becomes a “candidate” and “active”.
\State Let $\mathcal{P}_v = \langle d_v, ID_v \rangle$ be the priority of $v$. Priority is determined by degree, the node with the higher degree ($d_v$) has higher priority. The node's ID is used to break any ties.
\State $L_v \xleftarrow[]{}\mathcal{P}_v$ \Comment{$L_v$ is the current highest priority known to $v$.}
\State $N_v \xleftarrow[]{}\mathcal{P}_v$ \Comment{$N_v$ is the neighbor which informed about $L_v$.}
\State $v$ creates an arbitrary assignment of its neighbors based on its degree (from 1 to $d_v$) which are called $w_{v,1}, w_{v,2}, \cdots, w_{v,d_v}$ respectively.
\For {rounds $i=1$ to $\log d_v$} \label{loop: start}
    \If {$v$ is active}
        \State $v$ sends a “probe" message containing its priority $\mathcal{P}$ to its neighbors $w_{v,2^{i-1}}, \cdots, w_{v,\text{min}\{d_v,2^i-1\}}$.
       
    \EndIf
    \State Let $\chi_v$ be the possibly empty subset of $v's$ neighbors from which $v$ received messages in this round.
    \State Let $\Psi_v = \bigcup_1^i \chi_v$.
    \State Let $\mathcal{P}(\chi_v)$ be the set of $\mathcal{P}$s sent to $v$ by the members of $\chi_v$.
    \State Let $\mathcal{P}_u$ be the highest $\mathcal{P}$ in $\mathcal{P}(\chi_v)$.
    \If{$\mathcal{P}_u > L_v$} \label{line: if-clause}
        \State $L_v \xleftarrow[]{} \mathcal{P}_u$
        \State $v$ tells $N_v$ about $L_v = \mathcal{P}_u$, i.e., the highest $\mathcal{P}$ it has seen so far. \label{line: 17}
        \State $N_v \xleftarrow[]{} x$. \Comment{ $v$ remembers neighbor who told $v$ about $L_v$.}
        \State $v$ becomes "inactive" and "non-candidate".
    \EndIf
    \State $v$ tells every member of $\chi_v$ about $L_v$, i.e., the highest $\mathcal{P}$ it has seen so far. \label{line: 21}
    \EndFor \label{loop: for end}
    \If{$L_v =\mathcal{P}_v$}\label{line: leader}
        \State $v$ waits for $\log (d_v)$ rounds. If at the end of $\log (d_v)$ rounds, $L_v =\mathcal{P}_v$ then $v$ declares itself as leader and inform all the neighbors as well as exit the protocol.
    \EndIf
    \If{$v$ knows about the leader and $v$ is not the leader}
          \State Let $\Phi_v$ be the set of neighbors of $v$ to whom $v$ sent the messages before knowing about the leader.
          \State Let $\Psi_v=\Psi_v \bigcup \Phi_v$.
        \State $v$ informs $\Psi_v$ about the leader's ID and exit.
    \EndIf

\State All the nodes elect the same leader whose priority is the highest.
\end{algorithmic}
\end{algorithm}
\end{figure}

% \begin{lemma}
% Every graph with diameter two and $n$ nodes possess at least one node whose degree is more than $\lfloor \sqrt{n} \rfloor$.
% \end{lemma}
\begin{lemma}\label{lem: size}
Let $v$ be a node whose degree, $d_v$, is the highest among its neighbors and $\Delta$ is the maximum degree of the graph. There does not exist any diameter-two graph with $n$ nodes ($n>4$) such that $\Delta > d_v^2$.
\begin{proof}
For a node $v$, all nodes are at most $2$ hop distance away from $v$, since the diameter of the graph is $2$. Node $v$ has degree $d_v$ and its neighbors have degree at most $d_v$, by assumption. This gives an upper bound on $n$, that is, $n \leq d_v(d_v-1) + 1$, because each of the $d_v$ neighbors can have at most other $d_v-1$ neighbors (excluding $v$) each, and by the distance assumption there are no other nodes in the graph. Also, $\Delta$ can be at most $n-1$. Therefore, $\Delta < n < d_v^2+1$. Consequently, $d_v^2 > \Delta$. Hence, the lemma.

\end{proof}
\end{lemma}
\begin{remark}\label{rem:lb-on-max-degree}
It is clear that there does not exist any diameter-two network whose nodes are neither connected to $v$ nor its neighbor. Therefore, $d_v^2 \geq n$. This implies  $d_v \geq \sqrt{n}$. Hence, $\Delta \geq \sqrt{n}$.
\end{remark}

\begin{lemma}\label{lem: round_complexity}
Algorithm~\ref{alg : deterministic} solves the leader election in $O(\log \Delta)$ rounds, where $\Delta$ is the maximum degree of the graph.
\end{lemma}
\begin{proof}
A candidate node $v$ becomes the leader if its priority is the highest among its neighbors (Line \ref{line: leader}). From Lemma \ref{lem: size}, we know that $\Delta < d_v^2$. Therefore, the node $v$ with degree $d_v$ waits for $\log d_v$ rounds, in that time, the node with degree $\Delta$ inform about its priority to $v$ (if any) and $v$ becomes inactive. Otherwise, $v$ consider $d_v$ as $\Delta$ and inform all its neighbors about its election. $v$'s neighbor further conveys the message to all other nodes in $\log \Delta$ rounds. Therefore, the round complexity of the algorithm is $O(\log \Delta)$.
\end{proof}

% \begin{lemma}
% There exist at least one node in the diameter-two graph whose highest degree is more than $\sqrt{n}$.\manish{working...}
% \end{lemma}
% \begin{proof}
% Proof by contradiction. Let us suppose there exist a diameter-two graph such that every node has degree less than or equal to $\sqrt{n}$ and there exist a node $v$ with degree $\sqrt{n}$. Now, the node $v$ possess the degree $\sqrt{n}$ and all the neighbors' of $v$ can have degree at most $\sqrt{n}$ (including $v$ as a neighbor). If we further add the neighbors to the $v$'s neighbors' neighbor ($\Gamma\Gamma_v$) then the diameter will increase.
% This implies the total number of nodes in the graph are at most $(\sqrt{n}-1)^2+1$ 
% \end{proof}

For the message complexity analysis we adapt a couple of results from \cite{CPR20}, since our algorithm (Algorithm~\ref{alg : deterministic}) uses the similar approach to keep a node active. In particular, we use the Lemma~11 and Lemma~12 from \cite{CPR20}, which used ID of the nodes to take a decision on the "active" or "inactive" nodes whereas our algorithm uses priority (which depends on degree and ID). Hence, the results also applies to our algorithm.  %The another difference is, \cite{CPR20} runs the algorithm for $\log n$ rounds while Algorithm~\ref{alg : deterministic} runs for $\log \Delta$ rounds. Any node requires only $\log \Delta$ rounds to explore/communicate the neighbor, therefore, in our case we are optimizing the round. 
The following two lemmas are adapted from Lemma~11 and Lemma~12 in \cite{CPR20}.

\begin{lemma}[\cite{CPR20}]\label{lem: active_node}
At the end of the round $i$, there are at most $\frac{n}{2^{i}}$ “active” nodes.
\begin{proof}
Consider a node $v$ that is active at the end of round $i$. This implies that the if-clause of Line~\ref{line: if-clause} of Algorithm~\ref{alg : deterministic} has not so far been satisfied for $v$, which in turn implies that 
$\mathcal{P}_v > \mathcal{P}_{w_{v,j}}$ for $1\leq j \leq 2^i-1$, therefore none of $w_{v,1}, w_{v,2}, \dots, w_{v, 2^i-1}$ is active after round $i$. Thus, for every active node at the end of round $i$, there are at least $2^i-1$ inactive nodes. We call this set of inactive nodes, together with $v$ itself, the "kingdom" of $v$ after round $i$ i.e.,
$$KINGDOM_i(v) \overset{\mathrm{def}}{=} \{v\} \cup w_{v,1}, w_{v,2}, \dots, w_{v, 2^i-1} \text{ and } |KINGDOM_i(v)| = 2^i.$$
If we can show that these kingdoms are disjoint for two different active nodes, then we are done.\\
Proof by contradiction. Suppose not. Suppose there are two active nodes $u$ and $v$ such that
$$u \neq v \text{ and } KINGDOM_i(u) \cap KINGDOM_i(v)= \phi$$
(after some round $i$, $1 \leq i \leq \log n$). Let $x$ be such that
$x \in KINGDOM_i (u) \cap KINGDOM_i (v)$. Since an active node obviously cannot belong to the kingdom of another active node, this $x$ equals neither $u$ nor $v$, and therefore,
$$x \in \big\{w_{v,1}, w_{v,2}, \dots, w_{v, 2^i-1} \big\} \cap \big\{w_{u,1}, w_{u,2}, \dots, w_{u, 2^i-1} \big\}, $$
that is, both $u$ and $v$ have sent their respective probe-messages to $x$. Then it is straightforward to see that $x$ would not allow $u$ and $v$ to be active at the same time. Case-by-case analysis can be found in \cite{CPR20}.
\end{proof}
\end{lemma}

\begin{lemma}[\cite{CPR20}]\label{lem: round_message}
In round $i$, Algorithm~\ref{alg : deterministic} transmits at most $3n$ messages in the \textit{for} loop (from Line \ref{loop: start} to Line \ref{loop: for end}).
\end{lemma}
\begin{proof}
In round $i$, each active node sends exactly $2^i-1$ probe messages, and each probe-message generates at most two
responses (corresponding to Lines\ref{line: 17} and \ref{line: 21} of Algorithm~\ref{alg : deterministic}). Thus, in round $i$, each active node contributes to, directly or
indirectly, at most $3\cdot(2^i-1)$ messages. The result immediately follows from Lemma~\ref{lem: active_node}.
\end{proof}
% From the above discussion, we have the message complexity of the Algorithm~\ref{alg : deterministic} is $O(n \log \Delta)$  as shown in the Lemma~\ref{lem: message_complexity}.
\begin{lemma}\label{lem: message_complexity}
The message complexity of the Algorithm~\ref{alg : deterministic} is $O(n \log \Delta)$.
\end{lemma}
\begin{proof}
 Each round transmits at most $3n$ messages (Lemma~\ref{lem: round_message}) and the execution of the Algorithm~\ref{alg : deterministic} (from Line \ref{loop: start} to Line \ref{loop: for end}) takes place in $O(\log \Delta)$ rounds (Lemma~\ref{lem: round_complexity}). Further, leader informs about its election via $\Psi$ edges which are $O(n \log \Delta)$. Therefore, the total number of message transmitted throughout the execution are:
$3n \cdot O(\log \Delta) + O(n \log \Delta) = O(n \log \Delta)$.
\end{proof}

\noindent\textbf{Correctness of the Algorithm:} In this, we show that all the nodes agree on a leader and the leader is unique. First, we show that all the nodes agree on a leader. If a node $v$ is still a candidate node at the end of the first phase, then it must have both i) explored all its neighbors and ii) never encountered a priority higher than its own. Thus, it can declare itself leader after waiting $\log d_v$ rounds. Note that a waiting period of $\log d_v$ is enough because from Lemma~\ref{lem: size} we know that $\Delta < d_v^2$. This guarantees that the highest degree is made leader.
%Suppose a node $v$ does not know about the leader then $v$ might have come across a higher priority node or not, i.e, $v$ is either an "active" node or "inactive" node. In case if still "active", then $v$ must have explored all its neighbor, and now $v$ can declare itself as leader after waiting for $\log d_v$ rounds (see Lemma~\ref{lem: size}). Therefore, $v$ get to know about the leader. In other case, when $v$ is an "inactive" node then $v$ wait for the leader announcement and get to know about leader in $O(\log \Delta)$ rounds (Lemma~\ref{lem: round_complexity}). Therefore, all nodes agree on a leader.

Now, we show that the known leader is unique. If not, then suppose there exist two nodes $u$ and $v$ such that $u$ agrees on a leader $l_1$ and $v$ agree on a leader $l_2$. From algorithm~\ref{alg : deterministic}, $l_1$ should have the highest priority in its neighbors and similarly, $l_2$ should have the highest priority in its neighbors. Since it is a diameter two graph, therefore, there should be at least one node common among $l_1$ and $l_2$. Therefore, both the node can't have the highest priority among their neighbors, which is a contradiction. Therefore, we can say all the nodes agree on the unique leader.

From the above discussion, we conclude the following result.
\begin{theorem}\label{thm: deterministic}
There exists a deterministic (explicit) leader election algorithm for $n$-node anonymous networks with diameter two that sends $O(n \log \Delta)$ messages and terminates in $O(\log \Delta)$ rounds, where $\Delta$ is the maximum degree of the network.
\end{theorem}
\begin{remark}
The implicit deterministic leader election algorithm presented in \cite{CPR20} can be converted to an explicit leader election algorithm in the same way as done in Algorithm~\ref{alg : deterministic}. 
\end{remark}
% \manish{Is it appropriate to state following claim?\\}\sumathi{I think we should remove this.}
% Notice that the lower bound proved (in \cite{CPR20}) for any deterministic algorithm also holds for the explicit leader election on the same argument of the line. Thus, we have the more precise lower bound (supporting Algorithm~\ref{alg : deterministic}) in terms of highest degree of the graph, $\Delta$, as follows.

% \begin{theorem}\label{thm: lower_bound}
% There is no deterministic algorithm that solves
% leader election in every diameter-two network of $n$ nodes with $o(n \log \Delta)$ messages, where $\Delta$ is the highest degree of the graph.
% \end{theorem}
% \begin{remark}
% Theorem~\ref{thm: lower_bound}, does not contradict the result proved in \cite{CPR20} (Theorem 5) since there exist a diameter-two graph with $\Delta = O(n)$ . Our Theorem is a more precise way to state the result. 
% \end{remark}
% \begin{remark}
% From the Theorem~\ref{thm: deterministic} and Theorem~5 (from \cite{CPR20}), we deduce that every diameter-two network possess the $\Delta = n^c$, where $0<c\leq1$. 
% \end{remark}

\subsection{Broadcast Tree Formation}\label{sec:broadcast-tree}
In Algorithm~\ref{alg : deterministic}, the nodes agree on the leader explicitly. In this section, we exploit the edges used during the leader election algorithm (Algorithm~\ref{alg : deterministic}) and create a broadcast tree of height $O(\log \Delta)$ (Algorithm~\ref{alg : broadcast}). This also allows to reduce the message complexity. 
%To reduce the message complexity, we use only those edges which were used for the communication in Algorithm~\ref{alg : deterministic}. 
The process is simple. The leader, say $\ell$, initiates the flooding process by broadcasting its ID to its neighbors, forming the root of the tree $T$. All of its neighbors become a part of $T$. At any point in the algorithm, the leaves of $T$ do the following. Let $v$ be a leaf in $T$ in some round. In that round, $v$ sends its own ID to  the nodes in $\Psi_v$ (used in Algorithm~\ref{alg : deterministic}). Non tree nodes which receive an ID $v$ earlier become a part of $T$ with $v$ as its parent. If a non-tree node receives multiple messages, then it chooses the higher ID as its parent. The algorithm ends when all nodes have become a part of $T.$ Note that since only the leaves send out messages in each round and each node (except the root node, i.e., leader node) possess only one parent, we avoid the creation of cycles.

Let us now show some important lemmas which support the correctness of the algorithm.
In particular, Lemma \ref{lem: height} shows Algorithm~\ref{alg : broadcast} forms a tree of height $O(\log \Delta)$. The round complexity and message complexity of the Algorithm~\ref{alg : broadcast} is shown by Lemma~\ref{lem: time} and Lemma~\ref{lem: broadcast_message}, respectively. Finally, we conclude with message and round complexity as well as height of the tree in Theorem~\ref{thm: tree_formation}.

%when receives an ID. All the nodes consider the sender of the message as its parent node, therefore, leader becomes the root of the tree.

% \begin{algorithm}[H]
% \caption{\sc Tree Formation}\label{alg : broadcast}
% \begin{algorithmic}[1]
% \Require{Algorithm~\ref{alg : deterministic}.}
% \Ensure{Tree Structure.}

% \Statex

% \State The leader broadcast the message to all its neighbors.
% \State  Each neighboring node $v$ broadcast its ID via $\Psi_v$, iteratively.
% \If{node $u$ receives two ID simultenously}
%     \State $u$ consider that edge of tree which sent the higher ID and inform the higher node ID about the consideration of the edge.
% \Else 
%     \State Consider that edge of the tree which sent the ID first, and inform the corresponding node about the consideration. 
% \EndIf
% \State All the nodes consider the sender as parent. Therefore, leader become the root of the tree.
% \State The node $v$ exit the algorithm as $v$ sends its ID to the $\Psi_v$ nodes and consideration of the edge.
% \end{algorithmic}
% \end{algorithm}

\begin{algorithm}[h]
\caption{\sc Broadcast-Tree-Formation}\label{alg : broadcast}
\begin{algorithmic}[1]
\Require{A diameter-2 connected network graph $G$ in which each node possess unique ID.
%Leader $\ell$, for each $v$ in the network, $\Psi_v$ as created during the run of Algorithm~\ref{alg : deterministic}.
}
\Ensure{Tree Structure $T$.}

\State First run Algorithm~\ref{alg : deterministic} to elect the leader $\ell$. Each node also keeps track of its $\Psi_v$ (created during the course of the algorithm). 

\State $\ell$ becomes root of $T$. $\ell$ then broadcasts its ID as an invite to all its neighbors. And its neighbors become its children in $T$. \label{line: broadcast_begin}
\While{there are nodes outside of $T$} \Comment{Takes $O(\log \Delta)$ rounds.}

\State  Each node $v\in T$ broadcasts its ID to the nodes in  $\Psi_v$.
\If{node $u \notin T$ receives IDs from nodes in tree $T$}
    \State $u$ accept invitation based on the highest priority node, say $v$, and becomes $v$'s child in $T$. 
    %invitation arbitrarily from the ones it received, say $v$, and becomes $v$'s child in $T$.
\EndIf    %\Comment{A single iteration of while loops takes no more than one round (invite-reply)} \manish{Seems confusing!}
\EndWhile  \label{line: broadcast_end}
\end{algorithmic}
\end{algorithm}

% \begin{lemma}\label{lem: diameter}
% The height of the graph generated by all the $\Psi_v$ nodes is $O(\log \Delta)$.
% \end{lemma}
% \begin{proof}
% The proof is by contradiction. Let us suppose that the diameter of the graph generated by all the $\Psi_v$ nodes is more than $O(\log \Delta)$ then there exist two nodes, say $x$ and $y$, such that the distance between them is more than $O(\log \Delta)$. Therefore, any message sent from node $x$ can not be reached to node $y$ with in $O(\log \Delta)$ rounds. On the other hand, Algorithm~\ref{alg : deterministic} terminates in $O(\log \Delta)$ in such a way that there does not exist any unexplored node, which is contradictory to our assumption. Hence, the lemma.
% \end{proof}

\begin{lemma}\label{lem: time}
In $O(\log \Delta)$ rounds, all nodes are guaranteed to be part of the tree $T$.
\end{lemma}
\begin{proof}
%Let us analyze the leader election algorithm's (Algorithm~\ref{alg : deterministic}) occurrence in reverse order.
This is guaranteed from the use of leader election algorithm. Consider the graph $G'$ constructed as follows. Let $\ell$'s neighbors be its neighbors in $G$. For every other node $v\neq \ell$ its neighbors are $\psi_v$. Clearly, from Algorithm~\ref{alg : deterministic}, $G'$ is connected (as every node learns of $\ell$) and of diameter $O(\log \Delta)$. Let level $i$ denote all nodes that are at most $i$ hops away from $\ell$ in $G'$. We claim that in $i$ rounds, all nodes in level $i$ would become a part of the tree $T$. By using induction, this is clearly true for is $i=1$. Assuming it's true for $i$, nodes of $i+1$ would become part of the tree next as they are in the $\Psi_v$ of at least one node in level $i$ and thus would get an invite. And since the number of levels can be at most $O(\log \Delta)$, all nodes become part of $T$ in $O(\log \Delta)$ rounds.
% To see how, think of the leader election happening in reverse. That is, look the nodes that learned of the leader last in Algorithm~\ref{alg : deterministic}, we know from analysis that this can take no more than $O(\log \Delta)$ rounds. Clearly, these nodes form the leaves of the tree $T$. Notice that their paths to root are the same paths they got the information of the leader's ID in Algorithm~\ref{alg : deterministic}. And since this can take no more than $O(\log \Delta)$ rounds, this means that it becomes part of the tree in $O(\log \Delta)$ rounds. And since all nodes are guaranteed to have heard the leader's ID, this ensures that all nodes become part of $T$.
\end{proof}

\begin{lemma}\label{lem: height}
Algorithm~\ref{alg : broadcast} forms a tree of height $O(\log \Delta)$.
%The height of the tree generated by Algorithm~\ref{alg : broadcast} is $O(\log \Delta)$.
\end{lemma}
\begin{proof}
Since in each iteration of the while loop, the height of the tree is extended by at most 1 (that is by attaching children to the leaves of $T$). And since the algorithm ensures that all nodes have become a part of $T$ in $O(\log \Delta)$ iterations of the while loop, the height of $T$ can not be more than $O(\log \Delta)$. Notice that since each node accepts only one invite, there can be no creation of a cycle.
%In this, we show that the formed graph is connected and there does not exist any cycle, i.e., the graph is a tree. Furthermore, the height of the tree is $O(\log \Delta)$. 
% Algorithm~\ref{alg : broadcast} (from Line~\ref{line: broadcast_begin} to Line~\ref{line: broadcast_end}), broadcast the message via $\Psi_v$ edges of the graph. Therefore, a message broadcast from the leader would not take the time more than the diameter of the graph formed by $\Psi_v$ edges which is $O(\log \Delta)$ (Lemma~\ref{lem: diameter}). In this way, we consider the leader as root node which broadcast the message to its children (neighboring nodes) and these children broadcast the message via $\Psi_v$ edges, recursively. In a particular round, if a node $x$ receives the message from some nodes then considers the edge corresponding to node $y$ which informed earlier (in case of tie, the highest ID node) for the formation of a tree. In other words, $x$ considers $y$ as its parent. Thus, in a particular round, a node considers exactly one node as its parent and does not consider in later round. Furthermore, root node does not have any parent. Therefore, each node has a unique parent except the root node which forms a rooted tree. Since each node has the shortest distance from the root node. As a consequence, the distance of any node from the root node (leader node) remains unchanged as compared to the graph formed by $\Psi_v$ nodes. Therefore, we can say that the height of the tree is $O(\log \Delta)$. 
\end{proof}
%Note that Lemma~\ref{lem: height} immediately implies the following 
\begin{remark}
The diameter of the graph created by Algorithm~\ref{alg : broadcast} is $O(\log \Delta)$.
\end{remark}

\begin{lemma}\label{lem: broadcast_message}
Algorithm~\ref{alg : broadcast} takes $O(n \log \Delta)$ messages.
\end{lemma}
\begin{proof}
In Algorithm~\ref{alg : deterministic}, for every node $v$ communication takes place via $\Psi_v$ edges in $O(n \log \Delta)$ messages (Theorem~\ref{thm: deterministic}). In Algorithm~\ref{alg : broadcast} (from Line~\ref{line: broadcast_begin} to Line~\ref{line: broadcast_end}) communication also takes place via same edges ($\Psi_v$) for two times. Therefore, message complexity remain unchanged to $O(n \log \Delta)$.
\end{proof}
% \textbf{Correctness of the Algorithm~\ref{alg : broadcast}:} In this we show that the formed graph is connected and there does not exist any cycle in the graph.
% Algorithm~\ref{alg : broadcast}, communicate via $\Psi_v$ edges for a vertex $v$ and Algorithm~\ref{alg : deterministic} shows that the $\Psi_v$ connects the graph since the algorithm elects the leader explicitly via $\Psi_v$ edges. Therefore
% %there exists a path from the root (leader) node to any other node $v$ and there exist a unique path between root and other nodes. 

Thus, from the above discussion, we conclude the following result.
\begin{theorem}\label{thm: tree_formation}
There exists an algorithm which solve the broadcast problem in $O(n \log \Delta)$ messages and $O(\log \Delta)$ rounds which generate a tree of height $O(\log \Delta)$.
\end{theorem}
% \begin{remark}  
% Theorem~\ref{thm: deterministic} possess the round complexity $\Omega(\log n)$ \cite{CPR20}. Therefore, we deduce that for diameter two-networks $O(\log \Delta) = O(\log n)$.
% \end{remark}

%% file: conclusion.tex
\section{Conclusion and Future Work}\label{sec: conclusion}
We studied the leader election problem in diameter-two networks. We settled all the questions raised  by Chatterjee et al. \cite{CPR20} w.r.t. deterministic setting. Various open problems come to light due to our work. These are as follows:
\begin{enumerate}
    \item We presented an $O(\log \Delta)$-round and $O(n\log \Delta)$-message complexity algorithm for the explicit leader election. An interesting question is to reduce the round complexity to $O(1)$ while keeping the message complexity $O(n\log n)$?
    
    \item Tree formed by broadcast has height $O(\log \Delta)$. An interesting question rises whether this is optimal when the message and round complexity remain unchanged or constant height is possible.
    \item Is it possible to have a randomized algorithm (with high probability) with message complexity $O(n \log n)$ and constant round complexity without the knowledge of $n$?
    %\item Removing the assumption that nodes have their unique ID, keeping $n$ unknown. Is it possible for a randomized algorithm (with high probability) to reach at agreement on a leader with $\Tilde{O}(n)$ message and constant rounds?
   \item With or without the knowledge of $n$, what would be the complexity and lower bound (in deterministic setting) in the LOCAL model where nodes can communicate with arbitrary message size in a round?
   % \item Agreement is a weaker problem as compared to leader election, in which, each node has an input value and reaches on agreement w.r.t. one given input value. It would be interesting to see whether message complexity $\Omega(n \log n)$ is a tight bound w.r.t. agreement problem.
\end{enumerate}